\documentclass[paper,letterpaper]{ieice}
\usepackage[dvips]{graphicx}
\usepackage{theorem}
\usepackage{amsmath}
\usepackage{amssymb}
\usepackage{graphicx}
\usepackage{subfigure}
\usepackage{color}
\theorembodyfont{\rmfamily}    
\newtheorem{theorem}{Theorem}
\newtheorem{lemma}[theorem]{Lemma}
\newtheorem{definition}[theorem]{Definition}
\newtheorem{corollary}[theorem]{Corollary}

\newtheorem{proposition}[theorem]{Proposition}
\newtheorem{remark}[theorem]{Remark}
\newcommand{\qed}{\hfill$\square$}
\newcommand{\theoremfin}{\hfill$\blacksquare$}

\newenvironment{proof}{%
  \noindent{\em Proof.\ }}{%
  \hspace*{\fill}\qed \\
  \vspace{2ex}}

\newcommand{\mymid}{:~}

\newcommand{\textchange}[1]{#1}
\allowdisplaybreaks
\field{A}
\vol{93}
\no{9}
\title{Strongly Secure Privacy Amplification Cannot Be Obtained by
Encoder of Slepian-Wolf Code}

\titlenote{A part of this paper was presented at 2009 IEEE International
Symposium on Information Theory in Seoul, Korea. This paper is published 
in IEICE Trans. Fundamentals, vol.~93, no.~9, pp.~1650--1659, September 2010. }
\authorlist{
 \authorentry[shun-wata@is.tokushima-u.ac.jp]{Shun Watanabe}{m}{labelA}
 \authorentry[ryutaroh@rmatsumoto.org]{Ryutaroh Matsumoto}{m}{labelB}
 \authorentry[uyematsu@ieee.org]{Tomohiko Uyematsu}{e}{labelB}
}
\affiliate[labelA]{The author is with the Department of Information Science and Intelligent Systems, Tokushima University}
\affiliate[labelB]{The authors are with the Department of Communications
and Integrated Systems, Tokyo Institute of Technology }

\received{2009}{6}{5}
\revised{2010}{4}{2}

\begin{document}
\maketitle
\begin{summary}
Privacy amplification is a technique to distill
a secret key from a random variable by a function
so that the distilled key and eavesdropper's 
random variable are statistically independent. 
There are three kinds of security criteria for the 
key distilled by privacy amplification: 
the normalized divergence criterion, which is also
known as the weak security criterion,
the variational distance criterion, and 
the divergence criterion, which is also known as the 
strong security criterion. As a technique to distill
a secret key, it is known that the encoder of a 
Slepian-Wolf (the source coding with
full side-information at the decoder)
code can be used as a function
for privacy amplification if we employ the weak
security criterion. In this paper, we show that
the encoder of a Slepian-Wolf
code cannot be used
as a function for privacy amplification
if we employ the criteria other than
the weak one.
\end{summary}

\section{Introduction}
One of the fundamental problems in the cryptography is
the key agreement in which the legitimate parties,
usually referred to as Alice and Bob, share a secret 
key that is not known by the eavesdropper, usually referred
to as Eve. The problems on the key agreement in the
information theory was initiated by Maurer \cite{maurer:93},
and was also studied
by Ahlswede and Csisz\'ar \cite{ahlswede:93}.
In their formulation, Alice, Bob, and Eve have correlated
random variables $X^n$, $Y^n$, and $Z^n$ respectively. 
Then, Alice and Bob generate a secret key from $(X^n,Y^n)$
by using the public (authenticated) communication. 

Typically, a key agreement protocol consists of two
procedures: information reconciliation 
\cite{bennett:88,brassard:94} and privacy 
amplification \cite{bennett:88,bennett:95}.
The purpose of the information reconciliation
for Alice and Bob is to share an identical random
variable (with high probability) by using the public 
communication. Privacy amplification is a technique
to distill a secret key from the shared random variable by using
a function so that  Eve's knowledge
and the secret key are statistically independent. In order to focus on 
privacy amplification, we assume that Alice and Bob initially
share the random variables $X^n = Y^n$ in the rest of this paper.

As for the security of the secret key distilled by
privacy amplification, 
there are three kinds of security criteria: 
the normalized divergence criterion, which is also
known as the weak security criterion,
the variational distance criterion, and 
the divergence criterion, which is also known as the 
strong security criterion.
The normalized divergence criterion requires that
the key and Eve's knowledge $Z^n$ is 
(almost) statistically independent in the sense
that the divergence divided by $n$ (normalized divergence), 
or equivalently the mutual information 
divided by $n$, is negligible.
On the other hand,
the variational distance criterion and the
divergence criterion require that the
key and Eve's knowledge is (almost) statistically independent in
the sense that the variational distance and the divergence
are negligible respectively.

Traditionally, the normalized divergence criterion was employed
in the study of the key 
agreement (e.g.~\cite{maurer:93,ahlswede:93}).
\textchange{However, as pointed out by Maurer \cite{maurer:94b} 
and independently by Csisz\'ar \cite{csiszar:96},}
Eve might know a large part of the key even if the key
satisfies the normalized divergence criterion. Therefore, we should
use the divergence criterion. Indeed, recent studies
on the key agreement employ the divergence criterion
(e.g.~\cite{csiszar:04,naito:08}).

As one of techniques to distill a secret key, it is known that the encoder
of a Slepian-Wolf (the source coding with
full side-information at the decoder)
code \cite{slepian:73} can be used as
a function for privacy amplification.
For example, Ahlswede and Csisz\'ar 
used this technique implicitly \cite{ahlswede:93}, 
and Muramatsu 
used this technique explicitly \cite{muramatsu:06b}.

To describe the technique more precisely, 
let us consider the Slepian-Wolf code system such that
$X^n$ is the principal source and $Z^n$ is the 
side-information. Then, the output of the encoder, which
is regarded as the key, satisfies the normalized divergence
criterion if the coding rate of the code is close to
the compression limit and the decoding error probability
of the code is negligible. However, it has not been
clarified whether this technique can be used for
the divergence criterion. 

In this paper, we show that above mentioned technique cannot
be used for the divergence criterion.
Actually, we show that the divergence grows infinitely
in the order of $\sqrt{n}$, which suggests that
Eve might know a large part of the key.

\textchange{In order to show that the divergence grows in the order 
of $\sqrt{n}$, the second order converse coding
theorem of the Slepian-Wolf code system 
(Theorem \ref{proposition:sw-second-order-converse})
plays an important role.
The second order source coding (without side-information)
was studied by several authors \cite{hayashi:08,kontoyiannis:97,nomura:07},
and Theorem \ref{proposition:sw-second-order-converse} can be regarded
as a generalization of Hayashi's result \cite{hayashi:08} to the Slepian-Wolf
code system. During the process of
the review, Nomura and Matsushima published the 
result \cite{nomura:09b}
concerning the second order asymptotic of the Slepian-Wolf
code system (the source coding with full side-information
at the decoder)\footnote{The case in which both sources are
encoded was published in \cite{nomura:09}.}.} 

\textchange{
The difference between this paper and \cite{nomura:09b}
is summarized as follows.
Since the main purpose of this paper is to show that
Slepian-Wolf codes cannot be used as strongly secure 
privacy amplification, we only showed the converse
coding theorem. On the other hand, the main
purpose of \cite{nomura:09b} is to show the
second order coding theorem of the Slepian-Wolf code system,
and they showed both the direct and converse parts.
Although the independent and identically distributed (i.i.d.) source
is exclusively treated in this paper, the slightly wider class of sources,
i.e., the sources such that the conditional self-information
has an asymptotic normality, is treated in \cite{nomura:09b}. 
It should be noted that the  approach in \cite{nomura:09b} is different from
that in this paper and is similar to the approach in \cite{nomura:07}.}

Although the divergence criterion is the strongest notion
of security among the above mentioned three criteria, 
some researchers (eg.~\cite{renner:05d}) deem that the variational distance 
criterion is appropriate notion of security because
it matches with the universally composable security \cite{canetti:01},
which requires that the actual distribution of the key and
Eve's knowledge is indistinguishable from the ideal distribution
with which the key is uniformly distributed and
independent of Eve's knowledge. 
Therefore, it is worthwhile 
clarifying whether the key obtained by the above mentioned technique
satisfies the variational distance criterion or not.
In this paper, we show that the key obtained by the technique
does not satisfy the variational distance criterion.
Actually, we show that the variational distance 
converges to one (the maximum amount), which means that
the actual distribution and the ideal distribution are
completely distinguishable.

The results in this paper are also interesting from
the view point other than privacy amplification.
The above mentioned technique can be regarded as the 
Slepian-Wolf version of the folklore theorem
shown by Han \cite{han:05}. 
Recently, Hayashi \cite{hayashi:08} showed that
the folklore theorem does not hold if we employ
the variational distance criterion nor the divergence
criterion instead of the normalized divergence criterion.
Our results can be regarded as a generalization of 
Hayashi's results for the Slepian-Wolf code.

The rest of this paper is organized as follows:
In Section \ref{preliminaries}, we review the basic 
notations, privacy amplification, and the above
mentioned technique. 
In Section \ref{main-results-2}, we show our main results
concerning the divergence criterion and their proofs.
In Section \ref{main-results},
we show our main results concerning 
the variational distance criterion and their proofs.
In Section \ref{sec:conclusion}, we conclude the paper.

Finally, it should be noted that the results on
the divergence criterion and 
the variational distance criterion cannot be derived 
from each other, though a weak version of
the result on the divergence criterion, i.e.,
the fact that the divergence does not converge
to zero (Corollary \ref{corollary-of-main-theorem}), 
can be derived as a corollary of
the results on the variational distance criterion.
The weak version only suggests that Eve might know
a few bits about the key whose length grows
infinitely as $n$ goes to infinity, which is not 
a serious problem in practice.
On the other hand, the result in Section \ref{main-results-2}
suggests that Eve's knowledge about the key also grows infinitely
as the length of the key goes to infinity, which is a serious
problem in practice. Therefore, we need to treat both the divergence
criterion and the variational distance criterion separately.

\section{Preliminaries}
\label{preliminaries}
\subsection{Privacy Amplification}

In this section, we review the basic
notations related to privacy amplification.
Suppose that Alice and Bob have a random variable $X^n$
on ${\cal X}^n$, and Eve has a random variable
$Z^n$ on ${\cal Z}^n$, where
$(X^n, Z^n)$ are independently
identically distributed (i.i.d.) according to
the probability distribution $P_{XZ}$.
In this paper, we assume that 
${\cal X}$ and ${\cal Z}$ are finite sets.

Privacy amplification \cite{bennett:88,bennett:95}
is a technique to distill a secret key $S_n$ from 
$X^n$ by using a function 
\begin{eqnarray*}
f_n:{\cal X}^n \to {\cal M}_n = \{1,\ldots,M_n \}
\end{eqnarray*}
so that the key and Eve's information $Z^n$
are statistically independent and the key is
uniformly distributed on the key alphabet ${\cal M}_n$.
The joint probability distribution of the key and Eve's 
information is given by
\begin{eqnarray}
\label{eq:definition-of-psz}
P_{S_nZ^n}(s,z^n) = \sum_{x^n \in f_n^{-1}(s)}
 P_{X^n Z^n}(x^n, z^n)
\end{eqnarray}
for $(s,z^n) \in {\cal M}_n \times {\cal Z}^n$,
where we defined
$f_n^{-1}(s) = \{x^n \mymid f_n(x^n) = s \}$.

For probability distributions $P$ and $Q$
on ${\cal A}$, let 
\begin{eqnarray}
d(P,Q) &=& \frac{1}{2} \sum_{a \in {\cal A}}
  |P(a) - Q(a)| \\
  &=& \textchange{ P({\cal B}) - Q({\cal B}) }
\label{eq:second-def-of-variational-distance}
\end{eqnarray}
be the variational distance (divided by $2$) \cite{cover}, \textchange{where
${\cal B} \subset {\cal A}$ is a set satisfying
$P(a) \ge Q(a)$ for $a \in {\cal B}$ and
$P(a) \le Q(a)$ for $a \in {\cal A}\backslash {\cal B}$.}
Let
\begin{eqnarray*}
D(P \| Q) = \sum_{a \in {\cal A}} P(a) \log \frac{P(a)}{Q(a)}
\end{eqnarray*}
be the divergence \cite{cover}, 
where we take the base of the logarithm to be $e$
throughout the paper.
By using these two quantities, we introduce three kinds
of security criteria on privacy amplification.

\begin{definition}
\label{def:security-definition-1}
If a sequence of functions $\{f_n\}$ satisfies
\begin{eqnarray}
\label{eq:security-definition-1}
\lim_{n \to \infty} \frac{1}{n} D(f_n) = 0
\end{eqnarray}
for 
\begin{eqnarray*}
D(f_n) = D(P_{S_nZ^n} \| P_{U_n} \times P_{Z^n}),
\end{eqnarray*} 
then we define privacy amplification by 
$\{ f_n \}$ to be secure with respect to 
the {\em normalized divergence criterion}, where
$P_{U_n}$ is the uniform distribution on ${\cal M}_n$.
\theoremfin
\end{definition}

\begin{definition}
\label{def:security-definition-2}
If a sequence of functions $\{f_n\}$ satisfies
\begin{eqnarray}
\label{eq:security-definition-2}
\lim_{n \to \infty} \Delta(f_n) = 0
\end{eqnarray}
for 
\begin{eqnarray*}
\Delta(f_n) = d(P_{S_nZ^n}, P_U \times P_{Z^n}),
\end{eqnarray*}
then we define privacy amplification by 
$\{ f_n \}$ to be secure with respect to
the {\em variational distance criterion}.
\theoremfin
\end{definition}

\begin{definition}
\label{def:security-definition-3}
If a sequence of functions $\{f_n\}$ satisfies
\begin{eqnarray}
\label{eq:security-definition-3}
\lim_{n \to \infty} D(f_n) = 0,
\end{eqnarray}
then we define privacy amplification by 
$\{ f_n \}$ to be secure with respect to 
the {\em divergence criterion}.
\theoremfin
\end{definition}

We can show that Eq.~(\ref{eq:security-definition-3})
implies Eq.~(\ref{eq:security-definition-2}) by using
Pinsker's inequality \cite{cover}. 
We can also show that
Eq.~(\ref{eq:security-definition-2}) implies
Eq.~(\ref{eq:security-definition-1}) by using 
\cite[Lemma 1]{csiszar:04}.  

The security criteria in Definitions \ref{def:security-definition-1}
and \ref{def:security-definition-3} are equivalent to 
the weak security criterion and the strong security
criterion defined in \cite{maurer:00}.
The security criterion in Definition \ref{def:security-definition-2}
is widely used recently (eg.~\cite{renner:05d}) 
because it matches with the universally
composable security \cite{canetti:01}, which requires that
the actual distribution $P_{S_n Z^n}$ and the ideal distribution
$P_{U_n} \times P_{Z^n}$ are indistinguishable.
Although the divergence is also related to the distinguishability
between distributions,
the variational distance is directly related to
the distinguishability because the optimized average probability
of the correct discrimination is given by
\begin{eqnarray}
\lefteqn{ \frac{1}{2} \max_{{\cal A} \subset {\cal M}_n \times {\cal Z}^n}
 [ P_{S_n Z^n}({\cal A}) + P_{U_n Z^n}({\cal A}^c)] } \nonumber \\
 &=& \frac{1}{2}[ 1+ d(P_{S_n Z^n}, P_{U_n} \times P_{Z^n}) ],
\label{eq:distinguishability}
\end{eqnarray}
which is a straightforward consequence of the definition
of the variational distance \cite{cover},
where the superscript $c$ designate the complement
of the set.

\subsection{Privacy Amplification by an Encoder of Slepian-Wolf Code}
\label{subsec:PA-SW}

In this section, we explain the Slepian-Wolf code,
and then review the relation between privacy 
amplification and the Slepian-Wolf code.
We consider the Slepian-Wolf code system in which
$X^n$ is the principal source and $Z^n$ is the
side-information. The code system consists of
the encoder 
\begin{eqnarray*}
\phi_n: {\cal X}^n \to {\cal M}_n
\end{eqnarray*}
and the decoder
\begin{eqnarray*}
\psi_n: {\cal M}_n \times {\cal Z}^n \to {\cal X}^n,
\end{eqnarray*}
and we denote the code as $\Phi_n = (\phi_n,\psi_n)$.
The error probability of the code is defined as
\begin{eqnarray*}
\hspace{-3mm}
\varepsilon(\Phi_n) =
 P_{X^n Z^n}(\{ (x^n,z^n) \mymid \psi_n(\phi_n(x^n),z^n) \neq x^n \}).
\end{eqnarray*}

For any real number $R >0$, the rate
$R$ is said to be achievable if there exists a sequence of
codes $\{ \Phi_n \}$ that satisfies
\begin{eqnarray}
\limsup_{n \to \infty} \frac{1}{n} \log M_n \le R ~\mbox{ and }~
\lim_{n \to \infty} \varepsilon(\Phi_n) = 0.
\label{eq:definition-of-achievable}
\end{eqnarray}
Then, we define the compression limit as
\begin{eqnarray*}
R_f(X|Z) = \inf \{ R \mymid R \mbox{ is achievable} \}.
\end{eqnarray*}
It is well known that the compression limit coincides
with the conditional entropy \cite{slepian:73}, i.e.,
$R_f(X|Z) = H(X|Z)$.

If a sequence of codes $\{ \Phi_n \}$ satisfies 
\begin{eqnarray}
\label{eq:compression-limit-achieving-1}
\lim_{n \to \infty} \frac{1}{n} \log M_n = H(X|Z)
\end{eqnarray}
and
\begin{eqnarray}
\label{eq:compression-limit-achieving-2}
\lim_{n \to \infty} \varepsilon(\Phi_n) = 0,
\end{eqnarray}
then we call the sequence of codes $\{ \Phi_n \}$ 
{\em compression limit achieving codes}.
When $\{ \Phi_n \}$ satisfies 
Eq.~(\ref{eq:compression-limit-achieving-1}),
it should be noted that the error probability
depends on the second order rate 
$\frac{1}{\sqrt{n}} \log \frac{M_n}{e^{nH(X|Z)}}$.
For later use, we present the converse coding
theorem concerning the tradeoff between the error
probability and the second order rate. The theorem
is a Slepian-Wolf coding version of the result
on the second order asymptotic of the source
coding \cite{hayashi:08}.

\begin{theorem}
\label{proposition:sw-second-order-converse}
Let $b \in \mathbb{R}$ be an arbitrary real number.
For any code sequence $\{ \Phi_n \}$, if the error 
probability satisfies 
\begin{eqnarray}
\label{eq:second-order-error-prob}
\limsup_{n \to \infty} \varepsilon(\Phi_n) <
 1 - G\left(\frac{b}{\sigma}\right),
\end{eqnarray}
then the rate satisfies
\begin{eqnarray}
\label{eq:second-order-rate-bound}
\liminf_{n \to \infty} 
  \frac{1}{\sqrt{n}} \log \frac{M_n}{e^{n H(X|Z)}} \ge b,
\end{eqnarray}
where 
\begin{eqnarray*}
G(t) = \int_{- \infty}^t \frac{1}{\sqrt{2 \pi}} e^{-u^2/2} du
\end{eqnarray*}
is the cumulative distribution function of
the Gaussian distribution with mean $0$ and variance $1$,
and where we set
\begin{eqnarray}
\label{eq:variance}
\sigma^2 = \mbox{Var}\left[ \log \frac{1}{P_{X|Z}(X|Z)} \right].
\end{eqnarray}
\theoremfin
\end{theorem}

This theorem is a straight forward consequence
of the central limit theorem, and we show a proof
in \ref{appendix-0}.

\textchange{In order to show the relation between 
privacy amplification and the Slepian-Wolf code,
we consider the situation in which Alice and Bob
share the principal source $X^n$ and Eve has
the side-information $Z^n$. Then, Alice and Bob
use the encoder $\phi_n$ as a function for 
privacy amplification, and we  regard
the output $S_n = \phi_n(X^n)$ of the encoder
as a secret key.}
The following proposition states that 
the encoders of compression limit achieving codes
can be used as functions for privacy
amplification if we employ the normalized
divergence criterion.

\begin{proposition}
\label{SW-folklore-theorem}
If a sequence of codes $\{ \Phi_n = (\phi_n,\psi_n) \}$
satisfies Eqs.~(\ref{eq:compression-limit-achieving-1})
and (\ref{eq:compression-limit-achieving-2}),
then we have
\begin{eqnarray*}
\lim_{n \to \infty} \frac{1}{n} D(\phi_n) = 0.
\end{eqnarray*} 
\theoremfin
\end{proposition}

This proposition can be proved almost in a similar
manner to \cite[Theorem 1]{muramatsu:06b}.
\textchange{
The proposition can be also dirived as a
special case of Eq.~(\ref{eq:normalized-upper-bound}) 
in Remark \ref{remark:growth-rate}}.
Note that Proposition \ref{SW-folklore-theorem}
can be regarded as the Slepian-Wolf version
of the folklore theorem \cite{han:05}
(see also \cite[Theorem 2.6.4]{han:book}).

\section{Divergence Criterion}
\label{main-results-2}

\subsection{Statement of Result}

In this section, we show our main result
concerning the divergence criterion,
which is proved in Section \ref{subsec:proof-div}.
In Section \ref{subsec:PA-SW}, we showed that
the encoders of compression limit achieving codes
can be used as functions for secure privacy
amplification in the sense of 
the normalized divergence criterion. 
The following theorem
states that the divergence actually grows 
infinitely in the order of $\sqrt{n}$ (Eq.~(\ref{eq:root-n-diverge})),
which suggests that Eve might know a large part
of the key.

\begin{theorem}
\label{theorem-divergence-bound}
Suppose that a sequence of functions,
$f_n: {\cal X}^n \to \{1,\ldots,M_n \}$ for
$n = 1,2,\ldots$,
satisfies Eq.~(\ref{eq:compression-limit-achieving-1}), 
and let
\begin{eqnarray*}
b = \liminf_{n \to \infty} \frac{1}{\sqrt{n}}
  \log \frac{M_n}{e^{nH(X|Z)}}.
\end{eqnarray*}
Then, we have
\begin{eqnarray}
\label{eq:divergence-bound}
\liminf_{n \to \infty} \frac{1}{\sqrt{n}}
  D(f_n) 
 \ge \int_{- \infty}^{\frac{b}{\sigma}} (b - \sigma u) g(u) du,
\end{eqnarray}
where 
\begin{eqnarray*}
g(u) = \frac{1}{\sqrt{2 \pi}} e^{-u^2/2}
\end{eqnarray*}
is the density function of the Gaussian distribution
with mean $0$ and variance $1$,
and where $\sigma^2$ is the variance defined
in Eq.~(\ref{eq:variance}).
\theoremfin
\end{theorem}

The theorem can be regarded as a generalization
of \cite[Theorem 8]{hayashi:08} for the Slepian-Wolf
code.

Suppose that the Slepian-Wolf code sequence 
$\{ \Phi_n = (\phi_n,\psi_n) \}$ satisfies
Eq.~(\ref{eq:compression-limit-achieving-1}) and
\begin{eqnarray}
\label{eq:smaller-than-one}
\limsup_{n \to \infty} \varepsilon(\Phi_n) < 1.
\end{eqnarray}
Then, there exists a real number $b \in \mathbb{R}$
such that
\begin{eqnarray*}
\limsup_{n \to \infty} \varepsilon(\Phi_n) 
 < 1 - G\left( \frac{b}{\sigma} \right).
\end{eqnarray*}
Since the right hand side of 
Eq.~(\ref{eq:divergence-bound}) is
an increasing function of $b$ and is
positive for any $b \in \mathbb{R}$, 
Theorem \ref{proposition:sw-second-order-converse}
and Theorem \ref{theorem-divergence-bound} imply that
there exists a constant $K > 0$ and an integer $n_0$ such that
\begin{eqnarray}
\label{eq:root-n-diverge}
D(\phi_n) \ge \sqrt{n} K
\end{eqnarray}
for every $n \ge n_0$.

\begin{remark}
\label{remark:growth-rate}
\textchange{
When the Slepian-Wolf code sequence satisfies
Eq.~(\ref{eq:compression-limit-achieving-1}), i.e.,
the first order rate is equal to the compression limit, 
and also satisfies Eq.~(\ref{eq:smaller-than-one}), we showed above that
the divergence divided by $\sqrt{n}$ is lower bounded by a constant
asymptotically, i.e.,
the divergence grows infinitely in the order of $\sqrt{n}$.
On the other hand,
when the first order rate is strictly larger than the compression
limit, we can show that the normalized divergence is lower bounded by
a constant asymptotically (Eq.~(\ref{eq:normalized-lower-bound})), i.e.,
the divergence grows infinitely
in the order of $n$. 
Furthermore, for a given first order rate, we can show that
the lower bound on the normalized divergence can be
achieved by using encoders of Slepian-Wolf 
codes (Eq.~(\ref{eq:normalized-upper-bound})).
It should be noted that 
Proposition \ref{SW-folklore-theorem} 
can be derived as a special case of 
Eq.~(\ref{eq:normalized-upper-bound}).
Eqs~(\ref{eq:normalized-lower-bound}) and (\ref{eq:normalized-upper-bound})
are proved in Appendix \ref{appendix-proof-of-normalized}.}
 
\textchange{
Let
$f_n:{\cal X}^n \to \{1,\ldots,M_n \}$ for
$n=1,2,\ldots$, 
be a sequence of functions (not necessarily encoders
of Slepian-Wolf codes) that satisfies
\begin{eqnarray*}
R = \liminf_{n \to \infty} \frac{1}{n} \log M_n > H(X|Z).
\end{eqnarray*}
Then, the inequality
\begin{eqnarray}
\label{eq:normalized-lower-bound}
\liminf_{n \to \infty} \frac{1}{n} D(f_n)
 \ge R - H(X|Z)
\end{eqnarray}
implies that the divergence grows infinitely
in the order of $n$.}

\textchange{
Let $\{ \Phi_n = (\phi_n,\psi_n) \}$ be a 
sequence of Slepian-Wolf codes satisfying
Eq.~(\ref{eq:definition-of-achievable}).
Then, we have
\begin{eqnarray}
\label{eq:normalized-upper-bound}
\limsup_{n \to \infty} \frac{1}{n} D(\phi_n) \le R - H(X|Z).
\end{eqnarray}
}

\textchange{
Eqs.~(\ref{eq:normalized-lower-bound}) 
and (\ref{eq:normalized-upper-bound})
clarify the optimal trade-off between the (first order) rate
of the secret key and the normalized divergence.
Eq.~(\ref{eq:normalized-upper-bound}) also states that
the optimal trade-off can be achieved by encoders
of Slepian-Wolf codes.
Evaluation of the trade-off between the so-called equivocation rate,
which is essentially equivalent to the normalized divergence,
and the rate of transmitted message 
was well studied in the context of the
wire-tap channel in the literatures (e.g.~\cite{wyner:75,csiszar:78}).
\theoremfin }
\end{remark}

\subsection{Proof of Theorem \ref{theorem-divergence-bound}}
\label{subsec:proof-div}

In order to show a proof of 
Theorem \ref{theorem-divergence-bound}, we 
need the following lemma.

\begin{lemma}
\label{lemma-divergence-bound}
Let 
\begin{eqnarray*}
\lefteqn{ {\cal T}_n = 
 \left\{ (x^n,z^n) \mymid  \phantom{\frac{1}{\sigma \sqrt{n}}} \right. } \\
&& \left. \frac{1}{\sigma \sqrt{n}} 
  \left(\log \frac{1}{P_{X^n|Z^n}(x^n|z^n)} - nH(X|Z) \right)
  \le \frac{b}{\sigma} \right\}.
\end{eqnarray*}
Then, we have
\begin{eqnarray*}
\lefteqn{ H(S_n|Z^n) } \\
&\le& \sum_{(x^n,z^n) \in {\cal T}_n}
  P_{X^n Z^n}(x^n,z^n) 
   \log \frac{1}{P_{X^n|Z^n}(x^n|z^n)} \\
&& + P_{X^nZ^n}({\cal T}_n^c) [
 \log M_n - \log P_{X^n Z^n}({\cal T}_n^c) ].
\end{eqnarray*}
\theoremfin
\end{lemma}

\begin{proof}
Let 
\begin{eqnarray*}
{\cal M}_n^\prime = {\cal M}_n \cup {\cal X}^n, 
\end{eqnarray*}
and let 
$f_n^\prime: {\cal X}^n \times {\cal Z}^n \to {\cal M}_n^\prime$
be the function defined by
\begin{eqnarray*}
f_n^\prime(x^n,z^n) = \left\{
 \begin{array}{ll}
 f_n(x^n) & \mbox{if } (x^n,z^n) \notin {\cal T}_n \\
 x^n & \mbox{if } (x^n,z^n) \in {\cal T}_n
\end{array} \right..
\end{eqnarray*}
We set the random variable $S_n^\prime = f_n^\prime(X^n,Z^n)$.
Then, we have
\begin{eqnarray}
\lefteqn{ H(S_n^\prime|Z^n) } \nonumber \\
 &=& \sum_{(x^n, z^n) \in {\cal T}_n}
  P_{X^nZ^n}(x^n,z^n) \log \frac{1}{P_{X^n|Z^n}(x^n|z^n)} \nonumber \\
 && + \sum_{(s,z^n) \in {\cal M}_n \times {\cal Z}^n}
  P_{S_n^\prime Z^n}(s,z^n) \log \frac{1}{P_{S_n^\prime|Z^n}(s|z^n)}. 
\nonumber \\
 \label{eq:proof-lemma-div-bound-1}
\end{eqnarray}

By using the log-sum inequality \cite{cover},
we can upper bound the last term in Eq.~(\ref{eq:proof-lemma-div-bound-1}) as
\begin{eqnarray}
\lefteqn{ \sum_{(s,z^n) \in {\cal M}_n \times {\cal Z}^n}
  P_{S_n^\prime Z^n}(s,z^n) 
  \log \frac{P_{Z^n}(z^n)}{P_{S_n^\prime Z^n}(s,z^n)} } \nonumber \\
&\le& P_{X^nZ^n}({\cal T}_n^c) [
 \log M_n - \log P_{X^n Z^n}({\cal T}_n^c) ].
\label{eq:proof-lemma-div-bound-2}
\end{eqnarray}

Let $f_n^{\prime\prime}: {\cal M}_n^\prime \to {\cal M}_n$
be the function defined by
\begin{eqnarray*}
f_n^{\prime\prime}(s) = \left\{ \begin{array}{ll}
  s & \mbox{if } s \in {\cal M}_n \\
 f_n(s) & \mbox{if } s \in {\cal X}^n
\end{array} \right. .
\end{eqnarray*}
Then, we have $S_n = f_n^{\prime\prime}(S_n^\prime)$.
Since the conditional entropy does not increase
by a function \cite{cover}, by combining
Eqs.~(\ref{eq:proof-lemma-div-bound-1}) and
(\ref{eq:proof-lemma-div-bound-2}), we have
the assertion of the lemma.
\end{proof}

\noindent{\em Proof of Theorem \ref{theorem-divergence-bound} }

By using Lemma \ref{lemma-divergence-bound}, we have
\textchange{
\begin{eqnarray*}
\lefteqn{ \frac{1}{\sqrt{n}} D(f_n) } \\
 &=& \frac{1}{\sqrt{n}} \left[
  \log M_n - H(S_n|Z^n) \right] \\
 &\ge& \frac{1}{\sqrt{n}} \left[ \log M_n 
   \phantom{\frac{1}{P_{X^n|Z^n}(x^n|z^n)}} \right. \\
 && - \sum_{(x^n,z^n) \in {\cal T}_n}
  P_{X^n Z^n}(x^n,z^n) 
   \log \frac{1}{P_{X^n|Z^n}(x^n|z^n)} \\
 && \hspace{-7mm} \left. \phantom{\frac{1}{P_X}} - P_{X^nZ^n}({\cal T}_n^c) \{
 \log M_n - \log P_{X^n Z^n}({\cal T}_n^c) \} \right] \\
 &=& P_{X^nZ^n}({\cal T}_n) \frac{1}{\sqrt{n}} \log \frac{M_n}{e^{nH(X|Z)}} \\
 && - \sigma \sum_{(x^n,z^n) \in {\cal T}_n} P_{X^n Z^n}(x^n,z^n) \\
  && \frac{1}{\sigma \sqrt{n}} \left( \log \frac{1}{P_{X^n|Z^n}(x^n|z^n)}
   - n H(X|Z) \right) \\
 && + \frac{1}{\sqrt{n}} P_{X^n Z^n}({\cal T}_n^c) 
  \log P_{X^nZ^n}({\cal T}_n^c).
\end{eqnarray*}}
By taking the limit of both sides and using 
the central limit theorem with
respect to the cumulative distribution function
\begin{eqnarray*}
\hspace{-5mm} \Pr\left\{
 \frac{1}{\sigma \sqrt{n}} 
  \left( \log \frac{1}{P_{X^n|Z^n}(X^n|Z^n)} - nH(X|Z) \right)
  \le u \right\},
\end{eqnarray*}
we have
\begin{eqnarray*}
\lefteqn{ \liminf_{n \to \infty} \frac{1}{\sqrt{n}} D(f_n) } \\
&\ge& b G\left(\frac{b}{\sigma} \right) 
  - \sigma \int_{- \infty}^{\frac{b}{\sigma}} u g(u) du \\
&=& \int_{- \infty}^{\frac{b}{\sigma}} (b - \sigma u) g(u) du,
\end{eqnarray*} 
which completes the proof
\qed

\section{Variational Distance Criterion}
\label{main-results}

\subsection{Statement of Results}

In this section, we show our main results
concerning the variational distance criterion, which are proved in 
Sections \ref{subsec:proof-1} and \ref{subsec:proof-2}. First, we define
the quantity $\delta(P_{X^n Z^n})$ as follows.

\begin{definition}
Let $1 \le M_n \le |{\cal X}|^n$
be an integer, and 
${\cal C}_n = \{ {\cal C}_{z^n} \}_{z^n \in {\cal Z}^n}$
be a family of sets such that
each ${\cal C}_{z^n} \subset {\cal X}^n$
satisfies $|{\cal C}_{z^n}| = M_n$, where 
$|{\cal A}|$ means the cardinality of
a set ${\cal A}$.
We define the distribution $P_{{\cal C}_n}$
on ${\cal X}^n \times {\cal Z}^n$ as
\begin{eqnarray*}
P_{{\cal C}_n}(x^n,z^n) = \left\{
\begin{array}{ll}
\frac{1}{M_n} P_{Z^n}(z^n) & \mbox{if } x^n \in {\cal C}_{z^n} \\
0 & \mbox{else}
\end{array}
\right. .
\end{eqnarray*}
Then, we define
\begin{eqnarray}
\label{eq:definition-of-delta}
\delta(P_{X^n Z^n}) = 
  \min_{{\cal C}_n} d(P_{X^n Z^n}, P_{{\cal C}_n}),
\end{eqnarray}
where the minimization is taken over all possible
choices of ${\cal C}_n$ for arbitrary
$1 \le M_n \le |{\cal X}|^n$.
\theoremfin
\end{definition}

In Section \ref{subsec:PA-SW}, we showed that
the encoders of compression limit achieving codes
can be used as functions for secure privacy
amplification in the sense of 
the normalized divergence criterion. 
However, the following Theorem \ref{theorem-trade-off} shows 
a trade-off (with some exceptions) between 
the error probability
$\varepsilon(\Phi_n)$ and the security parameter
$\Delta(\phi_n)$ for any code $\Phi_n$. 

\begin{theorem}
\label{theorem-trade-off}
For arbitrary Slepian-Wolf code $\Phi_n = (\phi_n,\psi_n)$, we have
\begin{eqnarray*}
\varepsilon(\Phi_n) 
  + \Delta(\phi_n) \ge
\delta(P_{X^n Z^n}).
\end{eqnarray*}
\theoremfin
\end{theorem}

\begin{theorem}
\label{theorem-delta-limit}
If the variance $\sigma^2$ defined in Eq.~(\ref{eq:variance})
is positive, then we have
\begin{eqnarray}
\label{eq:delta-limit}
\lim_{n \to \infty} \delta(P_{X^n Z^n}) = 1.
\end{eqnarray}
\theoremfin
\end{theorem}

The combination of Theorems \ref{theorem-trade-off} and
\ref{theorem-delta-limit} states that we cannot 
use the encoders of any (good) Slepian-Wolf codes as
functions for secure privacy amplification
if we employ the variational distance criterion.

\begin{corollary}
\label{corollary:variational-distance}
For arbitrary Slepian-Wolf code $\Phi_n = (\phi_n,\psi_n)$, if
\begin{eqnarray*}
\lim_{n \to \infty} \varepsilon(\Phi_n) = 0,
\end{eqnarray*}
then we have
\begin{eqnarray*}
\lim_{n \to \infty} \Delta(\phi_n) = 1
\end{eqnarray*}
provided that $\sigma^2 > 0$. 
\theoremfin
\end{corollary}

From Eq.~(\ref{eq:distinguishability}), 
Corollary \ref{corollary:variational-distance} means that
the actual distribution $P_{S_n Z^n}$ and the
ideal distribution $P_{U_n} \times P_{Z^n}$ are
completely distinguishable asymptotically.

The combination of Pinsker's inequality and 
Corollary \ref{corollary:variational-distance} implies
the following corollary, which states
that the keys obtained by
the encoders of any (good) Slepian-Wolf codes 
do not satisfy the divergence criterion,
although we have shown stronger result (Theorem \ref{theorem-divergence-bound})
in Section \ref{main-results-2}.

\begin{corollary}
\label{corollary-of-main-theorem}
For any
sequence of Slepian-Wolf codes 
$\{ \Phi_n = (\phi_n, \psi_n) \}$ such that
$\lim_{n \to \infty} \varepsilon(\Phi_n) = 0$, we have
\begin{eqnarray*}
\liminf_{n \to \infty} D(\phi_n) \ge \frac{2}{\ln 2}
\end{eqnarray*}
provided that $\sigma^2 > 0$. 
\theoremfin
\end{corollary}

The corollary only suggests that
Eve might know a few bits about the long key, 
which is not a serious problem in practice.
On the other hand, the stronger result  
suggests that Eve's knowledge about the key grows infinitely
as the length of the key goes to infinity, which is a serious
problem in practice. 

Theorems \ref{theorem-trade-off} and
\ref{theorem-delta-limit} can be regarded as
a generalization of \cite[Theorem 4]{hayashi:08}
for the Slepian-Wolf code.
Therefore, Theorems \ref{theorem-trade-off} and
\ref{theorem-delta-limit} can also interpreted as that
the Slepian-Wolf version
of the folklore theorem does not hold for the
variational distance criterion.

\begin{remark}
For a distribution $P_{XZ}$ with
$\sigma = 0$,
we can easily show that 
$\delta(P_{X^n Z^n}) = 0$ for any $n$ by
taking ${\cal C}_{z^n}$ as the support of
$P_{X^n|Z^n}(x^n|z^n)$. 
\theoremfin
\end{remark}

\begin{remark}
\label{remark-trade-off}
It should be noted that Theorem \ref{theorem-trade-off}
holds not only for i.i.d.~random variables $(X^n,Z^n)$,
but also for any $(X^n, Z^n)$.
\theoremfin
\end{remark}

\subsection{Proof of Theorem \ref{theorem-trade-off}}
\label{subsec:proof-1}

Before we show a proof of Theorem \ref{theorem-trade-off},
we introduce the following lemma.

\begin{lemma}
\label{lemma:existence-of-injection-code}
For arbitrary code $\Phi_n = (\phi_n, \psi_n)$,
there exists a code $\Phi_n^\prime = (\phi_n, \psi_n^\prime)$
that satisfies
\begin{eqnarray}
\label{eq:error-prob-decrease}
\varepsilon(\Phi_n^\prime) \le \varepsilon(\Phi_n)
\end{eqnarray}
and
\begin{eqnarray}
\label{eq:injection-condition}
\phi_n(\psi_n^\prime(s, z^n)) = s~~~
\forall (s,z^n) \in {\cal M}_n \times {\cal Z}^n.
\end{eqnarray}
\theoremfin
\end{lemma}

\begin{proof}
We construct a decoder $\psi_n^\prime$ as follows.
If $\phi_n(\psi_n(s,z^n)) \neq s$, then we set
$\psi_n^\prime(s, z^n) = \tilde{x}^n$
for arbitrarily chosen 
$\tilde{x}^n \in \phi_n^{-1}(s)$.
Otherwise, we set $\psi_n^\prime(s,z^n) = \psi_n(s,z^n)$.
From the construction of this decoder, it is obvious that
the code $\Phi_n^\prime = (\phi_n,\psi_n^\prime)$ satisfies
Eqs.~(\ref{eq:error-prob-decrease}) 
and (\ref{eq:injection-condition}).
\end{proof}

\noindent{\em Proof of Theorem \ref{theorem-trade-off} }

From Lemma \ref{lemma:existence-of-injection-code},
it suffices to prove Theorem \ref{theorem-trade-off}
for codes satisfying Eq.~(\ref{eq:injection-condition}).
Therefore, we assume that a code $\Phi_n$ satisfies
Eq.~(\ref{eq:injection-condition}) in
the rest of this section.

By using the decoder $\psi_n$, we construct the map
\begin{eqnarray}
\overline{\psi}_n(s,z^n) = (\psi_n(s,z^n), z^n).
\end{eqnarray}
Since the decoder satisfies the condition 
in Eq.~(\ref{eq:injection-condition}), 
$\overline{\psi}_n$ is an injection map from
${\cal M}_n \times {\cal Z}^n$ into
${\cal X}^n \times {\cal Z}^n$.

For the extended code $\overline{\Phi}_n = (\phi_n,\overline{\psi}_n)$,
we define the error probability
\begin{eqnarray*}
\lefteqn{ \hspace{-3mm} \varepsilon(\overline{\Phi}_n) } \\
 && \hspace{-5mm} =
 P_{X^n Z^n}(\{ (x^n,z^n) \mymid \overline{\psi}_n(\phi_n(x^n),z^n) \neq
  (x^n,z^n) \}).
\end{eqnarray*}
Obviously, we have $\varepsilon(\overline{\Phi}_n) = \varepsilon(\Phi_n)$.

Next, we define the distribution 
$\overline{P_{U_n} \times P_{Z^n}}$, which is
the embedding of $P_{U_n} \times P_{Z^n}$
into ${\cal X}^n \times {\cal Z}^n$, as follows:
\begin{eqnarray*}
\overline{P_{U_n} \times P_{Z^n}}(x^n,z^n) =
 P_{U_n} \times P_{Z^n}(\overline{\psi}_n^{-1}(x^n,z^n))
\end{eqnarray*}
for $(x^n, z^n) \in \overline{\psi}_n({\cal M}_n \times {\cal Z}^n)$,
and $\overline{P_{U_n} \times P_{Z^n}}(x^n,z^n) = 0$ for 
other $(x^n,z^n)$. Similarly, we define
the distribution $\overline{P_{S_n Z^n}}$, which is
the embedding of $P_{S_n Z^n}$ into ${\cal X}^n \times {\cal Z}^n$.

Since the decoder $\psi_n$ satisfies Eq.~(\ref{eq:injection-condition}),
we have
\textchange{
\begin{eqnarray}
P_{X^n Z^n}(x^n,z^n) 
  &\le& \sum_{\tilde{x}^n \in \phi_n^{-1}(\phi_n(x^n))} 
 P_{X^n Z^n}(\tilde{x}^n,z^n) \nonumber \\
 &=& \overline{P_{S_nZ^n}}(x^n,z^n)
\label{eq:proof-1}
\end{eqnarray}
}
for $(x^n,z^n) \in \overline{\psi}_n({\cal M}_n \times {\cal Z}^n)$,
where the equality in Eq.~(\ref{eq:proof-1}) follows
from Eq.~(\ref{eq:definition-of-psz})
and the definition of $\overline{P_{S_n Z^n}}$. 
On the other hand, we have
\begin{eqnarray}
\label{eq:proof-2}
P_{X^n Z^n}(x^n,z^n) \ge \overline{P_{S_n Z^n}}(x^n,z^n) = 0
\end{eqnarray}
for $(x^n,z^n) \in ({\cal X}^n \times {\cal Z}^n) \backslash \overline{\psi}_n({\cal M}_n \times {\cal Z}^n)$. \textchange{Noting Eqs.~(\ref{eq:proof-1})
and (\ref{eq:proof-2}) and using 
Eq.~(\ref{eq:second-def-of-variational-distance}),} we have
\begin{eqnarray}
\lefteqn{ d(P_{X^n Z^n}, \overline{P_{S_n Z^n}}) } \nonumber \\
&=& 
P_{X^n Z^n}(({\cal X}^n \times {\cal Z}^n) \backslash \overline{\psi}_n({\cal M}_n \times {\cal Z}^n)).
\label{eq:proof-3}
\end{eqnarray}
By using Eq.~(\ref{eq:proof-3}), we can rewrite 
$\varepsilon(\overline{\Phi}_n)$ as
\begin{eqnarray*}
\varepsilon(\overline{\Phi}_n) &=&
  P_{X^n Z^n}(({\cal X}^n \times {\cal Z}^n) \backslash \overline{\psi}_n({\cal M}_n \times {\cal Z}^n)) \\
 &=& d(P_{X^n Z^n}, \overline{P_{S_n Z^n}}).
\end{eqnarray*}

Finally, from the definition of $\delta(P_{X^n Z^n})$ and 
the triangular inequality, we have
\begin{eqnarray*}
\delta(P_{X^n Z^n}) 
 &\le& d(P_{X^n Z^n}, \overline{P_{U_n} \times P_{Z^n}}) \\
 &\le& d(P_{X^n Z^n}, \overline{P_{S_n Z^n}}) \\
  && + d(\overline{P_{S_n Z^n}}, \overline{P_{U_n} \times P_{Z^n}}) \\
 &=& \varepsilon(\Phi_n) + \Delta(\phi_n),
\end{eqnarray*}
which completes the proof of
Theorem \ref{theorem-trade-off}. \qed 

\subsection{Proof of Theorem \ref{theorem-delta-limit}}
\label{subsec:proof-2}

Let $\{ {\cal C}_n \}$ be the sequence of
the families such that 
$d(P_{X^n Z^n},P_{{\cal C}_n}) = \delta(P_{X^n Z^n})$
for each $n$.
For arbitrary positive constant $b > 0$, 
we divide ${\cal X}^n \times {\cal Z}^n$ into
the following three subsets:
\begin{eqnarray*}
{\cal A}_+ &=& \{ (x^n, z^n) \mymid M_n^{-1} e^{b \sqrt{n}} 
   < P_{X^n|Z^n}(x^n|z^n) \}, \\
{\cal A}_- &=& \{ (x^n, z^n) \mymid P_{X^n|Z^n}(x^n|z^n) \le
    M_n^{-1} e^{- b \sqrt{n}} \},
\end{eqnarray*}
and ${\cal A}_0 = ({\cal X}^n \times {\cal Z}^n) \backslash ({\cal A}_+ \cup {\cal A}_-)$.
Let
\begin{eqnarray*}
\overline{{\cal C}}_n = \bigcup_{z^n \in {\cal Z}^n} 
  \{ (x^n, z^n) \mymid x^n \in {\cal C}_{z^n} \},
\end{eqnarray*}
which is the support of $P_{{\cal C}_n}$.

We bound $\delta(P_{X^n Z^n})$ as follows:
\textchange{
\begin{eqnarray*}
\lefteqn{ \delta(P_{X^n Z^n}) } \\
&=& \frac{1}{2} \left[ \sum_{(x^n,z^n) \in {\cal A}_+} 
   | P_{X^n Z^n}(x^n,z^n) - P_{{\cal C}_n}(x^n,z^n)| \right. \\
&& \hspace{-3mm}
   \left. + \sum_{(x^n,z^n) \in {\cal A}_- \cap \overline{{\cal C}}_n} 
   | P_{X^n Z^n}(x^n,z^n) - P_{{\cal C}_n}(x^n,z^n)| \right. \\
&& \hspace{-3mm}
   \left. + \sum_{(x^n,z^n) \in {\cal A}_- \backslash \overline{{\cal C}}_n} 
   | P_{X^n Z^n}(x^n,z^n) - P_{{\cal C}_n}(x^n,z^n)| \right. \\
&& \hspace{-3mm}
   \left. + \sum_{(x^n,z^n) \in {\cal A}_0} 
   | P_{X^n Z^n}(x^n,z^n) - P_{{\cal C}_n}(x^n,z^n)| \right] \\
&\ge& \frac{1}{2}[ (P_{X^n Z^n}({\cal A}_+) - P_{{\cal C}_n}({\cal A}_+) ) \\
&& + (P_{{\cal C}_n}({\cal A}_-) 
      - P_{X^n Z^n}({\cal A}_- \cap \overline{{\cal C}}_n )) \\
&&      + P_{X^n Z^n}({\cal A}_- \backslash \overline{{\cal C}}_n)  \\
&& + (P_{{\cal C}_n}({\cal A}_0) - P_{X^n Z^n}({\cal A}_0) )] \\
&=& 1 - ( P_{{\cal C}_n}({\cal A}_+) 
   + P_{X^n Z^n}({\cal A}_- \cap \overline{{\cal C}}_n) \\
&&   + P_{X^n Z^n}({\cal A}_0) ),
\end{eqnarray*}
where, at the inequality, we used the relation
\begin{eqnarray*}
\lefteqn{ \sum_{a \in {\cal B}} |P(a) - Q(a)| } \\
&\ge& \max[P({\cal B}) - Q({\cal B}), Q({\cal B}) - P({\cal B})]
\end{eqnarray*}
for any distributions $P$ and $Q$ on ${\cal A} \supset {\cal B}$,
and we also used the facts
$P_{{\cal C}_n}({\cal A}_-\cap \overline{{\cal C}}_n) = P_{{\cal C}_n}({\cal A}_-)$
and 
$P_{{\cal C}_n}({\cal A}_-\backslash \overline{{\cal C}}_n) = 0$.
}

We use the following inequalities
\begin{eqnarray}
\label{eq:bound-1}
P_{{\cal C}_n}({\cal A}_+) &\le& e^{- b \sqrt{n}}, \\
P_{X^n Z^n}({\cal A}_- \cap \overline{{\cal C}}_n) 
  &\le& e^{- b \sqrt{n}},
\label{eq:bound-2}
\end{eqnarray}
and 
\begin{eqnarray}
\label{eq:bound-3}
P_{X^n Z^n}({\cal A}_0) \le \frac{2b}{\sqrt{2 \pi \sigma}} 
   + \frac{2 C_1}{\sqrt{n}} \left(\frac{\rho}{\sigma} \right)^3,
\end{eqnarray}
where $C_1$ is a constant that does not depend on $n$ and
$\rho$ is the third moment of 
$- \log P_{X|Z}(X|Z)$. We will prove these inequalities
in Appendices \ref{appendix-1}, 
\ref{appendix-2}, and \ref{appendix-3} respectively.

From Eqs.~(\ref{eq:bound-1})--(\ref{eq:bound-3})
and the fact that $b > 0$ is arbitrary, we have
\begin{eqnarray*}
\lefteqn{ \liminf_{n \to \infty} \delta(P_{X^n Z^n}) } \\
 &\ge& 1 - \limsup_{b \to 0} \limsup_{n \to \infty} [
  P_{{\cal C}_n}({\cal A}_+)  + \\
&& P_{X^n Z^n}({\cal A}_- \cap \overline{{\cal C}}_n) 
   + P_{X^n Z^n}({\cal A}_0) ] \\
&=& 1.
\end{eqnarray*}
Since the variational distance (divided by $2$) is smaller than
$1$, we have the statement of theorem. \qed




\section{Conclusion}
\label{sec:conclusion}

In this paper, we showed that the encoders of
(good) Slepian-Wolf codes cannot be used as 
functions for secure privacy amplification
in the sense of the variational distance criterion
nor the divergence criterion. 
The consequence of our results is that
we must use privacy amplification
not based on the Slepian-Wolf code
(e.g.~\cite{bennett:95,maurer:00,csiszar:04,renner:05b})
if we want to employ the criteria other 
than the weak one (the normalized divergence criterion).

\section*{Acknowledgment}
The authors would like to thank Dr.~Jun Muramatsu
for comments. The first author also would like to thank
Prof.~Yasutada Oohama for his support.
We also thank the anonymous reviewers for their
constructive comments and suggestions.
This research is partly supported by the Japan
Society of Promotion of Science under 
Grants-in-Aid No.~00197137 and
by Grant-in-Aid for Young Scientists (Start-up):
KAKENHI 21860064.

\appendix
\section{Proof of Theorem \ref{proposition:sw-second-order-converse}}
\label{appendix-0}

In order to show a proof of 
Theorem \ref{proposition:sw-second-order-converse},
we need the following lemma \cite{miyake:95} (see also \cite{han:book}).

\begin{lemma}
\label{lemma-sw-converse}
For any Slepian-Wolf code $\Phi_n = (\phi_n, \psi_n)$, we have
\begin{eqnarray*}
\lefteqn{ \varepsilon(\Phi_n)} \\
& \ge&  P_{X^nZ^n}\left(\left\{
 (x^n,z^n) \mymid \phantom{\log \frac{1}{P_{X^n|Z^n}(x^n|z^n)}} \right. \right. \\
 && \left. \left. \log \frac{1}{P_{X^n|Z^n}(x^n|z^n)} \ge \alpha_n
  \right\} \right) 
  - M_n e^{- \alpha_n}, 
\end{eqnarray*}
where $\alpha_n$ is arbitrary real number. 
\theoremfin
\end{lemma}

For arbitrarily fixed $\gamma > 0$,
suppose that there exists a code sequence 
$\{ \Phi_n \}$ that satisfies 
Eq.~(\ref{eq:second-order-error-prob}) and
\begin{eqnarray*}
\liminf_{n \to \infty} \frac{1}{\sqrt{n}} \log \frac{M_n}{e^{nH(X|Z)}}
 \le b - 2 \gamma.
\end{eqnarray*}
Then, there exists a increasing sequence 
$\{ n_i \}_{i=1}^\infty$ such that
\begin{eqnarray*}
\frac{1}{\sqrt{n_i}} \log \frac{M_{n_i}}{e^{n_i H(X|Z)}} \le b - \gamma 
\end{eqnarray*}
for every $i$.

By using Lemma \ref{lemma-sw-converse}
for $\alpha_{n_i} = \sqrt{n_i} b + n_i H(X|Z)$,
we have
\begin{eqnarray*}
\lefteqn{ \varepsilon(\Phi_{n_i})} \\
& \ge& P_{X^{n_i}Z^{n_i}}\left(\left\{
 (x^{n_i},z^{n_i}) \mymid  \phantom{\frac{1}{\sigma}} \right. \right.\\
&& \left. \left. \frac{1}{\sigma \sqrt{n_i}} \left( 
  \log \frac{1}{P_{X^{n_i}|Z^{n_i}}(x^{n_i}|z^{n_i})} - n_i H(X|Z) \right)  \right. \right. \\
&& \left. \left. \ge \frac{b}{\sigma}
  \right\} \right) 
   - e^{- \gamma \sqrt{n_i}}
\end{eqnarray*}
for every $i$.
By using the central limit theorem, we have
\begin{eqnarray*}
\limsup_{n \to \infty} \varepsilon(\Phi_n) 
&\ge& \limsup_{i \to \infty} \varepsilon(\Phi_{n_i}) \\
&\ge& 
  1 - G\left(\frac{b}{\sigma}\right),
\end{eqnarray*}
which contradicts Eq.~(\ref{eq:second-order-error-prob}).
Therefore, if the code sequence $\{ \Phi_n \}$
satisfies Eq.~(\ref{eq:second-order-error-prob}),
then it satisfies
\begin{eqnarray*}
\liminf_{n \to \infty} \frac{1}{\sqrt{n}} \log \frac{M_n}{e^{nH(X|Z)}}
 > b - 2 \gamma.
\end{eqnarray*}
Since $\gamma >0$ is arbitrary, we 
have the assertion of the theorem. \qed

\section{}
\subsection{Proof of Eq.~(\ref{eq:bound-1})}
\label{appendix-1}

From the definitions of $P_{{\cal C}_n}$ and ${\cal A}_+$, we have
\begin{eqnarray*}
P_{{\cal C}_n}({\cal A}_+)
 &\le& \sum_{(x^n,z^n) \in {\cal A}_+} \frac{1}{M_n} P_{Z^n}(z^n) \\
 &\le& \sum_{(x^n,z^n) \in {\cal A}_+} P_{X^n Z^n}(x^n,z^n) e^{- b \sqrt{n}} \\
  &\le& e^{- b \sqrt{n}}.
\end{eqnarray*}
\subsection{Proof of Eq.~(\ref{eq:bound-2})}
\label{appendix-2}

From the definitions of $P_{{\cal C}_n}$ and ${\cal A}_-$, we have
\begin{eqnarray*}
\lefteqn{ P_{X^n Z^n}({\cal A}_- \cap \overline{{\cal C}}_n) } \\
 &=& \sum_{(x^n,z^n) \in {\cal A}_- \cap \overline{{\cal C}}_n} 
  P_{X^n Z^n}(x^n,z^n) \\
 &\le& \sum_{(x^n,z^n) \in {\cal A}_- \cap \overline{{\cal C}}_n}
  \frac{1}{M_n} P_{Z^n}(z^n) e^{- b \sqrt{n}} \\
 &=& \sum_{(x^n,z^n) \in {\cal A}_- \cap \overline{{\cal C}}_n}
  P_{{\cal C}_n}(x^n,z^n) e^{- b \sqrt{n}} \\
 &\le& e^{- b \sqrt{n}}.
\end{eqnarray*}

\subsection{Proof of Eq.~(\ref{eq:bound-3})}
\label{appendix-3}

To simplify the notation, we introduce the random variable
\begin{eqnarray*}
W_n = \sum_{i = 1}^n \log \frac{1}{P_{X|Z}(X_i|Z_i)}
\end{eqnarray*}
for 
$(X^n,Z^n) = ((X_1,Z_1),\ldots, (X_n, Z_n))$.
Then, we can rewrite
the left hand side of Eq.~(\ref{eq:bound-3}) as
\begin{eqnarray*}
\lefteqn{P_{X^n Z^n}({\cal A}_0) } \\
 &=& \Pr\{ \log M_n - b \sqrt{n} \le W_n < \log M_n + b \sqrt{n} \} \\
 &=& \Pr\left\{ \frac{\log M_n - n H(X|Z)}{\sigma \sqrt{n}} 
   - \frac{b}{\sigma} \right. \\
 &&  \le \frac{W_n - n H(X|Z)}{\sigma \sqrt{n}} \\
 && \left. < 
 \frac{ \log M_n - n H(X|Z)}{ \sigma \sqrt{n}} + \frac{b}{\sigma} \right\} \\
 &=& \Pr\left\{ \frac{W_n - n H(X|Z)}{\sigma \sqrt{n}} \right. \\
 && \left. < 
 \frac{ \log M_n - n H(X|Z)}{ \sigma \sqrt{n}} + \frac{b}{\sigma} \right\} \\
 &-&  \Pr\left\{ \frac{W_n - n H(X|Z)}{\sigma \sqrt{n}} \right. \\
 && \left.  < 
 \frac{ \log M_n - n H(X|Z)}{ \sigma \sqrt{n}} - \frac{b}{\sigma} \right\}.
\end{eqnarray*}

By using the central limit theorem \cite[Corollary 6]{rao:book},
we have
\begin{eqnarray*}
\lefteqn{ \hspace{-3mm} \Pr\left\{ \frac{W_n - n H(X|Z)}{\sigma \sqrt{n}} < 
 \frac{ \log M_n - n H(X|Z)}{ \sigma \sqrt{n}} + \frac{b}{\sigma} \right\} } \\
&\le& 
  \int_{- \infty}^{\frac{\log M_n - n H(X|Z)}{\sigma \sqrt{n}} + \frac{b}{\sigma}} g(u) du + \frac{C_1}{\sqrt{n}}
  \left( \frac{\rho}{\sigma} \right)^3
\end{eqnarray*}
and 
\begin{eqnarray*}
\lefteqn{ \hspace{-3mm} \Pr\left\{ \frac{W_n - n H(X|Z)}{\sigma \sqrt{n}} < 
 \frac{ \log M_n - n H(X|Z)}{ \sigma \sqrt{n}} \textchange{-} \frac{b}{\sigma} \right\} } \\
&\ge& 
  \int_{- \infty}^{\frac{\log M_n - n H(X|Z)}{\sigma \sqrt{n}} - \frac{b}{\sigma}} g(u) du - \frac{C_1}{\sqrt{n}}
  \left( \frac{\rho}{\sigma} \right)^3.
\end{eqnarray*}
Hence, we have
\begin{eqnarray*}
\hspace{-7mm} P_{X^n Z^n}({\cal A}_0) \le 
 \int_{\frac{\log M_n - n H(X|Z)}{\sigma \sqrt{n}} - \frac{b}{\sigma}}^{\frac{\log M_n - n H(X|Z)}{\sigma \sqrt{n}} + \frac{b}{\sigma}} 
 g(u) du + \frac{2 C_1}{\sqrt{n}} 
 \left(\frac{\rho}{\sigma}\right)^3.
\end{eqnarray*}
Since the interval of the integral is $\frac{2 b}{\sigma}$ and
the height of $g(u)$ is lower than $1$, we have
Eq.~(\ref{eq:bound-3}).

\subsection{Proofs of Eqs.~(\ref{eq:normalized-lower-bound}) 
and (\ref{eq:normalized-upper-bound})}
\label{appendix-proof-of-normalized}

\textchange{
Eq.~(\ref{eq:normalized-lower-bound}) is derived by the inequality
\begin{eqnarray*}
D(f_n) &=& \log M_n - H(S_n|Z^n) \\
 &\ge& \log M_n - H(X^n|Z^n) \\
 &=& \log M_n - nH(X|Z).
\end{eqnarray*}
On the other hand, Eq.~(\ref{eq:normalized-upper-bound})
is derived by the inequality
\begin{eqnarray}
\hspace{-5mm} D(\phi_n) &=& \log M_n - H(X^n|Z^n) \nonumber \\
  && \hspace{-3mm} + [ H(X^n|Z^n) - H(S_n|Z^n)] \nonumber \\
\hspace{-5mm} &=&  \log M_n - n H(X|Z) + H(X^n|S_n,Z^n) \label{eq:normalized-div-bound-1} \\
\hspace{-5mm} &\le& \log M_n - nH(X|Z) \nonumber \\
 && \hspace{-5mm} + n \varepsilon(\Phi_n) \log |{\cal X}|
 + h(\varepsilon(\Phi_n)),
\label{eq:normalized-div-bound-2}
\end{eqnarray}
where we used the fact that $S_n$ is a function of $X^n$
in Eq.~(\ref{eq:normalized-div-bound-1}), 
we used Fano's inequality \cite{cover}
in Eq.~(\ref{eq:normalized-div-bound-2}), and $h(\cdot)$ is the binary
entropy function \cite{cover}.}




\end{document}